\documentclass[a4paper,onecolumn,11pt]{quantumarticle}
\pdfoutput=1
\usepackage[utf8]{inputenc}
\usepackage[english]{babel}
\usepackage[T1]{fontenc}
\usepackage{amsmath}
\usepackage{hyperref}

\usepackage{multirow}
\usepackage{tikz}
\usepackage{quantikz}
\usepackage{subcaption}
\usepackage{lipsum}
\usepackage{pgfplots}
\usetikzlibrary{patterns}

\usepackage{amsthm}
\newtheorem{theorem}{Theorem}
\usepackage[numbers,sort&compress]{natbib}

\usepackage[ruled,vlined]{algorithm2e}

\begin{document}

\title{A parallel and distributed fixed-point quantum search algorithm for solving SAT problems}

\author{He Wang}
\email{15584191817@163.com}
\author{Jinyang Yao}
\affiliation{Information Engineering University, 62 Kexue Avenue, Zhengzhou, 450001, Henan, China}

\maketitle

\begin{abstract}
 The Boolean satisfiability (SAT) problem is of fundamental importance in computer science and many application domains. For Grover’s algorithm, solving the SAT problem requires $\mathcal{O}(\sqrt{2^n})$ queries—where $n$ denotes the number of logic variables in the problem. However, Grover’s algorithm suffers from the Soufflé problem: specifically, when the number of solutions is unknown, terminating the algorithm too early or too late leads to a significant reduction in the probability of obtaining a solution. In this paper, we propose a parallel fixed-point (PFP) search algorithm to solve the SAT problem. By exploiting entanglement, each clause in the conjunctive normal form (CNF) formula can be processed independently, leading to a significant reduction in circuit depth. We also discuss how to perform the algorithm in distributed manner. These make the PFPS algorithm particularly suitable for the noisy intermediate-scale quantum (NISQ) era.
\end{abstract}

\section{Introduction}

Quantum computing, an emerging technological frontier, is gradually transitioning from theoretical paradigms to practical implementations, and demonstrating transformative potential in computational capabilities. A key milestone in this shift is humanity’s entry into the NISQ era \cite{Preskill2018quantumcomputingin}—a phase defined by quantum devices equipped with hundreds to thousands of qubits. While noise and limited coherence times remain unresolved in this era, these quantum systems can already execute small-scale operational tasks. The construction of large-scale universal quantum computers remains a significant challenge today, primarily hindered by the inherent noise in quantum systems and the limited depth of quantum circuits. Against this backdrop, within the NISQ era, the development of novel quantum algorithms and models with enhanced physical realizability has emerged as a compelling and promising research direction. 

The distributed quantum computing (DQC) constitutes one such feasible and valuable topic. DQC has been explored via a variety of methodologies and perspectives. As established in existing research, three general approaches are commonly adopted for distributed quantum computing scenarios. The first approach involves directly decomposing the quantum circuit associated with a target computational problem into multiple quantum sub-circuits; however, this requires the implementation of quantum communication protocols (such as quantum teleportation) to coordinate interactions between these sub-circuits \cite{lin2024parallel}. A typical illustration of this approach is the distributed Shor's algorithm presented in \cite{yimsiriwattana2004distributed}, though a key limitation is the substantial number of teleportation operations it necessitates. The second approach entails generating multiple local solutions using quantum algorithms analogous to the original problem-specific algorithm, and then synthesizing these local results to determine the final solution to the overall problem. For instance, this strategy is employed in the distributed quantum phase estimation method in \cite{li2017application}. The third approach, which has been proposed more recently, focuses on decomposing the Boolean function to be computed into a collection of subfunctions \cite{avron2021quantum}; the solution to the original problem is then derived by computing all or a selected subset of these subfunctions, as exemplified in \cite{tan2022distributed,qiu2024distributed,zhou2024distributed}.

SAT solving stands as a cornerstone problem in the realm of classical computing. For any given propositional logic formula, its primary function is to determine if there exist truth value assignments to propositional variables that can make the formula true. SAT boasts extensive applications across diverse domains, ranging from theorem proving, model checking, and circuit design, to name just a few. Beyond its real-world utility, SAT holds a pivotal role in computational and complexity theories. This is due to two key attributes: first, it is an NP-complete problem; second, a large number of other computational problems can be reduced to SAT. Developing efficient algorithms for solving SAT, therefore, holds far-reaching significance for a wide array of computer science domains and even fields outside of computing. The DPLL algorithm serves as one of the most prevalent classical algorithms for SAT solving which exhibits a worst-case bound of $ \mathcal{O}(2^n)$, where $n$ denotes the count of propositional variables within the target formula \cite{davis1962machine}. 

Grover's search algorithm is regarded as one of the powerful algorithms for solving the classical SAT problem \cite{grover1996fast}. This is because it exhibits an impressive query complexity of $\mathcal{O}(\sqrt{2^n})$ when searching through an n-bit classical unordered database. However, when the number of solutions is unknown in advance, Grover's algorithm encounters the well-known soufflé problem—where stopping the algorithm either too early or too late leads to a significant decrease in the probability of obtaining a solution. Some early improved algorithms can guarantee solving the soufflé problem, but only at the cost of losing the quadratic speedup \cite{grover2005fixed}. Thus, of greater appeal are the later-proposed fixed-point quantum search algorithms that retain the ability of quadratic speedup \cite{mizel2009critically,yoder2014fixed}. Another common approach is to employ quantum counting to derive the (approximate) count of target elements prior to implementing Grover’s algorithm \cite{brassard1998quantum}. However, to obtain the relatively accurate number of solutions, the time cost of this method may be higher than that of the aforementioned fixed-point search method.

In this paper, we design a parallel fixed-point (PFP) quantum search algorithm for solving SAT problems \cite{mizel2009critically,lin2024parallel}. Our algorithm first enables parallel processing of clauses, which significantly reduces the Oracle's runtime. Additionally, we discuss how to perform the algorithm in a distributed manner. Consequently, the total runtime of the algorithm can be markedly decreased by more than a factor of $\mathcal{O}(m)$, where $m$ denotes the number of clauses. Second, distributed processing is particularly suitable for the present NISQ era—this is because the number of available logical qubits in each quantum computer is currently small, insufficient to support large-scale computing tasks. Finally, this algorithm can avoid the soufflé problem while ensuring that it preserves the quantum characteristics of $\mathcal{O}(\sqrt{2^n})$. 

The remainder of the paper is structured as follows. In Sec. \ref{Preliminaries}, we first provide an introduction to the SAT problem, followed by an overview of Grover’s algorithm. In Sec.\ref{PFP search algorithm}, we elaborate on the algorithm in detail, prove its correctness, and conduct numerical simulations. Finally, we conclude with a summary in Sec.\ref{Summary}.

\section{Preliminaries}\label{Preliminaries}

\subsection{SAT problems}

At the core of a SAT problem lie $n$ Boolean variables, designated as the set  $\{x_j\}$ with $ j $ ranging from $1$ to $n$. Each variable $x_j$ is restricted to two truth values: true, which is represented by either $0$ or the $+$ symbol, and false, indicated by  $1$ or the $-$ symbol. A key building block in SAT logic is the literal $l_j$ —this component can take one of two forms: it is either the original variable $x_j$ or the negation of $x_j$, symbolized as $\overline{x_j}$ .
When focusing on $k$-SAT, a specialized subset of SAT, each clause $ C_i$ is constructed by combining exactly $k$ literals using the logical OR operator ( $\lor$ ). A full $k$-SAT instance is defined by a Boolean formula $F$ structured in conjunctive normal form (CNF). In this CNF structure,  $m$ individual clauses are linked together in sequence through the logical AND operator ( $\land $), resulting in the formula:
\begin{equation}
F = C_1 \land C_2 \land \cdots \land C_m
\end{equation}
As an example, consider 3-SAT (where $k = 3$): a typical clause $C_1$ might be structured as $C_1 = \overline{x_1} \lor x_3 \lor \overline{x_5}$. A clause is deemed satisfied if at least one of its literals evaluates to true. For the entire SAT instance to qualify as satisfiable, there must exist a specific assignment of truth values to the $n$ Boolean variables such that all $m$ clauses within formula $F$  are satisfied at the same time.
In essence, the SAT problem is classified as a decision problem. Its central objective is to determine whether a given SAT instance (i.e., a particular CNF formula $F$) is satisfiable (meaning there is at least one valid variable assignment that works) or unsatisfiable (with no such valid assignment existing).

\subsection{Grover's algorithm}

The initial step of Grover’s algorithm involves constructing a coherent superposition state with uniform amplitude distributed across all $N=2^n$ computational basis states. This state is formulated as $|\psi_0\rangle = \frac{1}{\sqrt{N}} \sum_{k=0}^{N-1} |k\rangle=\cos\theta/2 |\psi^*_\perp\rangle+\sin\theta/2 |\psi^*\rangle$, where the basis state $|k\rangle$ acts as the quantum encoding of the $k$-th classical data entry within the target database. $|\psi^*\rangle$ is a uniform superposition of target states, $|\psi^*\rangle=\frac{1}{\sqrt{M}}\sum_{k: f(k)=1} |k\rangle$, and $|\psi^*_\perp\rangle$ is its orthonormal complementary. We define $\sin\theta=\sqrt{\frac{M}{N}} $, where $M$ is the number of solutions.
Thereafter, the algorithm enters an iterative loop that alternately implements two core quantum operators: the Grover diffusion operator, denoted $G = -e^{i\pi |\psi_0\rangle\langle\psi_0|}$, and the oracle operator, defined as $O = e^{i Q \pi}$. In this formulation, $Q = \sum_{k: f(k)=1} |k\rangle\langle k|$ functions as a projection operator that maps onto the subspace spanned by all solution states. Within the two-dimensional Hilbert space $\{|\psi^*\rangle,|\psi^*_\perp\rangle\}$, we can also define $A=|\psi^*\rangle\langle\psi^*_\perp|+|\psi^*_\perp\rangle\langle\psi^*|$ and $B=i|\psi^*\rangle\langle\psi^*_\perp|-i|\psi^*_\perp\rangle\langle\psi^*|$, whose actions are analogous to those of Pauli operators in qubit space. The query function $f$ is formally defined such that $f(k) = 1$ if the classical data $k$ satisfies the problem’s solution criteria, and $f(k) = 0$ in all other cases. For example, $f$ can be a CNF formula.  

After $T = \frac{\pi}{4}\sqrt{\frac{N}{M}}$ iterative cycles of alternating $G$ and $O$ applications, the evolved quantum state converges to a close approximation of the target solution state $|\psi^*\rangle$. The algorithm’s time complexity (equivalent to its query complexity, given the dominance of oracle calls) is quantified as $T = \mathcal{O}(\sqrt{N})$. This represents a quadratic speedup relative to classical algorithms for unstructured search, which typically require $\mathcal{O}(N)$ queries to guarantee finding a solution. However, the challenge here lies in the fact that when the number of solutions is unknown in advance—a scenario that is typically the case—it becomes difficult to determine the exact number of iterations, leading to the soufflé problem. Mizel proposed a measurement-driven quantum search algorithm that ensures both quadratic speedup and avoidance of the soufflé problem \cite{mizel2009critically}. In this paper, our goal is to parallelize this algorithm to reduce its total runtime, and further implement it in a distributed manner—adapting it to the limitation of the small number of logical qubits available in a single quantum computer during the NISQ era.

\section{PFP search algorithm}\label{PFP search algorithm}

\subsection{Algorithm description}

The PFP search algorithm requires the use of three quantum registers and one classical register, respectively. These three quantum registers are the variable register, the formula register, and the control qubit register. The variable register contains different variable qubit groups, whose definition will be provided below. The formula register consists of one qubit that represents the calculation result of the formula. The control qubit register is composed of one control qubit. The classical register is used to store the measurement results of the control qubits. To enable parallel processing in the algorithm, we adopt the same strategy as that in \cite{lin2024parallel}: if a variable $x_j$ appears $r_j$ times across the clauses, we introduce an additional $r_j-1$ auxiliary qubits, ensuring that distinct qubits (the original $x_j$ and the auxiliary ones) appear in different clauses. This thus forms a variable qubit group $\{x_j^0, x_j^1, \ldots, x_j^{r_j-1}\}$. To ensure that all qubits in this variable qubit group yield consistent results during the final measurement, we initialize the entire qubit group into a GHZ state $\frac{1}{\sqrt{2}}(|000...\rangle+|111...\rangle)$. The additional independent variable qubits are initialized to $|+\rangle$. The algorithm proceeds as follows.

\begin{algorithm}
  \SetAlgoLined
  \caption{PFP search algorithm}
  \KwIn{GHZ states corresponding to different variable qubit groups; The independent variable qubits initialized to the \(|+\rangle\) state; The formula qubit initialized to the \(|0\rangle\) state; The control qubit initialized to the \(|1\rangle\) state.}
  \KwOut{A solution that satisfies the formula.}
  \For{each iteration}{
      Apply controlled-oracle $O_{ctrl}(\phi_t)$;\\
      Apply controlled-diffuser $G_{ctrl}$;\\
      Measure the control qubit,\\
      \If{ The measurement result is $0$}{Break;}
    \Else{Continue to the next iteration;}   
  }
  \Return{The control qubit collapses to $|0\rangle$. The formula qubit collapses to $|0\rangle$. The variable qubits are in a uniform superposition state of the solutions.}
\end{algorithm}

To construct controlled-oracle, we first build a circuit for each clause. The circuit is shown in Fig.\ref{fig:ClauseCircuit}. It should be noted that the qubits in each clause are exclusive. $R_i$ is selected according to the following rules \cite{fernandes2019using,lin2024parallel},

\begin{align}
R_i=\left\{\begin{aligned}
&X  \quad\text{if $l_i=x_i$}\\
&I  \quad\text{if $l_i=\overline{x_i}$}
\end{aligned}\right.\end{align}

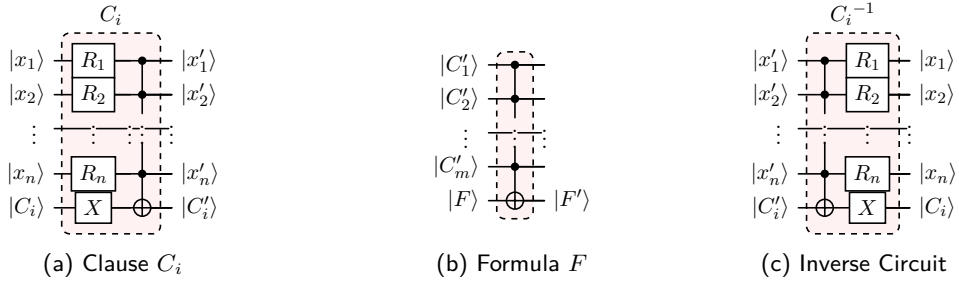
\begin{figure}
\centering
\begin{subfigure}[b]{0.4\linewidth}
\centering
\scalebox{0.75}{
\begin{quantikz}[row sep={6mm,between origins}, column sep=2mm]
\lstick{\ket{x_1}} & [1mm] \gate{R_1} \gategroup[5,steps=3,style={dashed, rounded corners,fill=pink!20, inner xsep=1pt, inner ysep=1.5pt}, background,label style={label position=above,anchor=south,yshift=-0.2cm}]{{\sc $C_i$}}&\qw  & [-2mm]\ctrl{1} & [1mm]\qw \rstick{\ket{x_1'}} \\
\lstick{\ket{x_2}} & \gate{R_2}& \qw & \ctrl{1} & \qw \qw \rstick{\ket{x_2'}} \\ 
\hspace{-7mm} \vdots & \vdots & \vdots & \vdots & \vdots \\[2mm]
\lstick{\ket{x_n}} & \gate{R_n}& \qw& \ctrl{1} \vqw{-1} & \qw \qw \rstick{\ket{x_n'}} \\
\lstick{\ket{C_i}} & \gate{X} & \qw & \targ{} & \qw \qw \rstick{\ket{C_i'}}
\end{quantikz}
}
\caption{Clause $C_i$} \label{fig:ClauseCircuit}
\end{subfigure}%
\hfill
\begin{subfigure}[b]{0.28\linewidth}
\centering
\scalebox{0.75}{
\begin{quantikz}[row sep={6mm,between origins}, column sep=2mm]
\lstick{$\ket{C_1'}$} & [1mm]\ctrl{1} \gategroup[5,steps=1,style={dashed, rounded corners,fill=pink!20, inner xsep=1pt, inner ysep=1pt}, background,label style={label position=above,anchor=south,yshift=-0.2cm}]{{\sc }} & [1mm]\qw \\
\lstick{$\ket{C_2'}$} & \ctrl{1} & \qw \\
\hspace{-8mm} \vdots & \vdots & \vdots \\
\lstick{$\ket{C_m'}$} & \ctrl{1} \vqw{-1} & \qw \\
\lstick{\ket{F}} & \targ{} & \qw \rstick{\ket{{F}'}}
\end{quantikz}
}
\\ 
\vspace{2mm}
\caption{Formula $F$} \label{fig:FormulaCircuit}
\end{subfigure}%
\hfill
\begin{subfigure}[b]{0.3\linewidth}
\centering
\scalebox{0.75}{
\begin{quantikz}[row sep={6mm,between origins}, column sep=2mm]
\lstick{\ket{x_1'}} & [1mm] \ctrl{1} \gategroup[5,steps=2,style={dashed, rounded corners,fill=pink!20, inner xsep=1pt, inner ysep=1.5pt}, background,label style={label position=above,anchor=south,yshift=-0.2cm}]{{\sc ${C_i}^{-1}$}} & \gate{R_1} & [1mm]\qw \rstick{\ket{x_1}} \\
\lstick{\ket{x_2'}} & \ctrl{1} & \gate{R_2} & \qw \qw \rstick{\ket{x_2}} \\ 
\hspace{-7mm} \vdots & \vdots & \vdots & \vdots \\[2mm]
\lstick{\ket{x_n'}} & \ctrl{1} \vqw{-1} & \gate{R_n} & \qw \qw \rstick{\ket{x_n}} \\
\lstick{\ket{C_i'}} & \targ{} & \gate{X} & \qw \qw \rstick{\ket{C_i}}
\end{quantikz}
}
\caption{Inverse Circuit} \label{fig:ClauseInverseCircuit}
\end{subfigure}%
\caption{Quantum circuit for clause, formula and the inverse of clause.}
\end{figure}

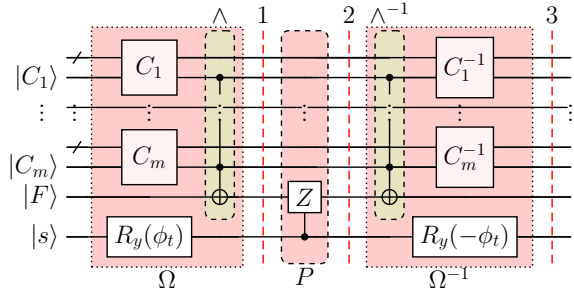
\begin{figure}[tb]\centering\scalebox{0.66}{\begin{quantikz}[row sep={4mm,between origins}, column sep=4mm, font=\Large] & \qw \qwbundle{} & \gate[wires=2, style={fill=pink!20}][11mm]{C_1} \gategroup[7,steps=2,style={dotted,fill=red!20, inner xsep=5pt, inner ysep=2pt}, background,label style={label position=below,anchor=south,yshift=-0.75cm}]{{\sc $\Omega$}} & \qw \gategroup[6,steps=1,style={dashed, rounded corners,fill=olive!20, inner xsep=0pt, inner ysep=0pt}, background,label style={label position=above,anchor=south,yshift=-0.2cm}]{{\sc $\wedge$}} & [-2mm]\qw \slice{1} & [6mm]\qw \gategroup[7,steps=1,style={dashed, rounded corners,fill=red!20, inner xsep=0pt, inner ysep=0pt}, background,label style={label position=below,anchor=south,yshift=-0.75cm}]{{\sc $P$}} & [-5mm] \qw \slice{2} & [9mm] \qw \gategroup[7,steps=2,style={dotted,fill=red!20, inner xsep=5pt, inner ysep=2pt}, background,label style={label position=below,anchor=south,yshift=-0.75cm}]{{\sc $\Omega^{-1}$}} \qw\gategroup[6,steps=1,style={dashed, rounded corners,fill=olive!20, inner xsep=0pt, inner ysep=0pt, thick}, background, label style={label position=above,anchor=south,yshift=-0.2cm}]{{\sc $\wedge^{-1}$}}  & [-1mm]\gate[wires=2, style={fill=pink!20}]{C_1^{-1}} & \qw \slice{3} & [2mm]\qw \\ \lstick{\ket{C_1}} & \qw & \qw & \ctrl{1} & \qw & \qw & \qw & \ctrl{1} & \qw & \qw & \qw \\[2mm]   \hspace{-10mm} \vdots & \hspace{-5mm} \vdots & \vdots & \vdots & & \vdots & & \vdots & \hspace{-2mm} \vdots & & \vdots \\[4mm] & \qw \qwbundle{} & \gate[wires=2, style={fill=pink!20}][11mm]{C_m} & \qw & \qw & \qw & \qw & \qw &[-1mm]\gate[wires=2, style={fill=pink!20}]{C_m^{-1}} & \qw & \qw \\ \lstick{\ket{C_m}} & \qw & \qw & \ctrl{1} \vqw{-2} & \qw & \qw & \qw & \ctrl{1} \vqw{-2} & \qw & \qw & \qw \\[2mm] \lstick{\ket{F}} & \qw & \qw & \targ{} & \qw & \gate[wires=1]{Z} & \qw & \targ{} & \qw & \qw & \qw\\[4mm]\lstick{\ket{s}} & \qw & \gate{R_y(\phi_t)} & \qw & \qw & \ctrl{-1} & \qw & \qw & \gate{R_y(-\phi_t)} & \qw & \qw \end{quantikz}}\caption{Quantum circuit for controlled parallel oracle.}\label{fig:ParallelOracleControlled}\end{figure}

The circuit of the controlled oracle comprises three components: the $\Omega$ module, the $P$ module, and the $\Omega^{-1}$ module. In the $\Omega$ module, first, the values of the clauses are computed via clause circuits. The control qubit undergoes a rotation about the y-axis with a rotation angle of $\phi_t$. Subsequently, the computation results $|C_i\rangle$ from all clause circuits are processed through a multi-controlled NOT gate to perform a conjunction (see Fig.\ref{fig:FormulaCircuit} and olive zone in Fig.\ref{fig:ParallelOracleControlled}), so as to compute the result of the formula. Module P primarily implements a controlled-Z gate to invert the phase of the target state. The $\Omega^{-1}$ module is the inverse circuit of the $\Omega$ module and serves to restore the amplitudes of all quantum states in Fig.\ref{fig:ParallelOracleControlled}. Unlike conventional circuits, which can only compute all clauses sequentially, this circuit computes all clauses in parallel simultaneously. For a circuit with $m$ clauses, the conventional oracle runs sequentially with a time complexity of $\mathcal{O}(m)$, whereas the parallel oracle operates with a time complexity of only $\mathcal{O}(1)$.

\begin{figure}[tb]\centering\scalebox{0.7}{\begin{quantikz}[row sep={5mm,between origins}, column sep=2mm, font=\Large]\lstick{$\ket{x_{1[0]}}$} & \qw & [6mm]\qw \slice{1} & [6mm]\ctrl{1} \gategroup[2,steps=1,style={dotted,fill=olive!20, inner xsep=1pt, inner ysep=-1pt}, background,label style={label position=above,anchor=south,yshift=-2mm}]{} & [-2mm]\qw \slice{2} & [6mm]\gate{H} \gategroup[10,steps=5,style={dotted,fill=red!20, inner xsep=1pt, inner ysep=1pt}, background,label style={label position=above,anchor=south,yshift=-2mm}]{} & \gate{X} & [1mm]\ctrl{2} \gategroup[10,steps=1,style={dashed, rounded corners,fill=olive!20, inner xsep=0pt, inner ysep=0pt}, background,label style={label position=above,anchor=south,yshift=-0.2cm}]{} & [1mm]\gate{X} & \gate{H} & [-2mm]\qw \slice{3} & [6mm]\ctrl{1} \gategroup[2,steps=1,style={dotted,fill=olive!20, inner xsep=1pt, inner ysep=-1pt}, background,label style={label position=above,anchor=south,yshift=-2mm}]{} & \qw & [4mm]\qw \\\lstick{$\ket{x_{1[i\neq0]}}$} & \qw \qwbundle{r_1-1} & \qw & \targ{} & \qw & \qw & \qw & \qw & \qw & \qw & \qw & \targ{} & \qw & \qw \\[2mm] \lstick{$\ket{x_{2[0]}}$} & \qw & \qw & \ctrl{1} \gategroup[2,steps=1,style={dotted,fill=olive!20, inner xsep=1pt, inner ysep=-1pt}, background,label style={label position=above,anchor=south,yshift=-2mm}]{} & \qw & \gate{H} & \gate{X} & \ctrl{2} & \gate{X} & \gate{H} & \qw & \ctrl{1} \gategroup[2,steps=1,style={dotted,fill=olive!20, inner xsep=1pt, inner ysep=-1pt}, background,label style={label position=above,anchor=south,yshift=-2mm}]{} & \qw & \qw \\\lstick{$\ket{x_{2[i\neq0]}}$} & \qw \qwbundle{r_2-1} & \qw & \targ{} & \qw & \qw & \qw & \qw & \qw & \qw & \qw & \targ{} & \qw & \qw \\[-1mm] \hspace{-12mm} \vdots & \vdots & & \vdots & & \vdots & \vdots & \vdots & \vdots & \vdots & & \vdots & & \vdots \\[2mm]  \lstick{$\ket{x_{n-1[0]}}$} & \qw & \qw & \ctrl{1} \gategroup[2,steps=1,style={dotted,fill=olive!20, inner xsep=1pt, inner ysep=-1pt}, background,label style={label position=above,anchor=south,yshift=-2mm}]{} & \qw & \gate{H} & \gate{X} & \ctrl{2} \vqw{-1} & \gate{X} & \gate{H} & \qw & \ctrl{1} \gategroup[2,steps=1,style={dotted,fill=olive!20, inner xsep=1pt, inner ysep=-1pt}, background,label style={label position=above,anchor=south,yshift=-2mm}]{} & \qw & \qw \\\lstick{$\ket{v_{n-1[i\neq0]}}$} & \qw \qwbundle{r_{n-1}-1} & \qw & \targ{} & \qw & \qw & \qw & \qw & \qw & \qw & \qw & \targ{} & \qw & \qw \\[2mm]  \lstick{$\ket{x_{n[0]}}$} & \qw & \qw & \ctrl{1} \gategroup[2,steps=1,style={dotted,fill=olive!20, inner xsep=1pt, inner ysep=-1pt}, background,label style={label position=above,anchor=south,yshift=-2mm}]{} & \qw & \gate{H} & \gate{X} & \gate{Z} & \gate{X} & \gate{H} & \qw & \ctrl{1} \gategroup[2,steps=1,style={dotted,fill=olive!20, inner xsep=1pt, inner ysep=-1pt}, background,label style={label position=above,anchor=south,yshift=-2mm}]{} & \qw & \qw \\\lstick{$\ket{x_{n[i\neq0]}}$} & \qw \qwbundle{r_{n}-1} & \qw & \targ{} & \qw & \qw & \qw & \qw & \qw & \qw & \qw & \targ{} & \qw & \qw \\\lstick{$\ket{s}$} & \qw & \qw & \qw & \qw & \qw & \qw & \ctrl{-2} & \qw & \qw & \qw & \qw & \qw & \qw \end{quantikz}}\caption{The parallel diffuser scheme with additional control qubit. $x_{j[0]}$ denotes the 0th qubit in the variable qubit group $x_j$, while $x_{j[i\neq0]}$ denotes all other qubits in the variable qubit group $x_j$ except the 0th one. } \label{fig:ParallelDiffuserWithCtrl}\end{figure}
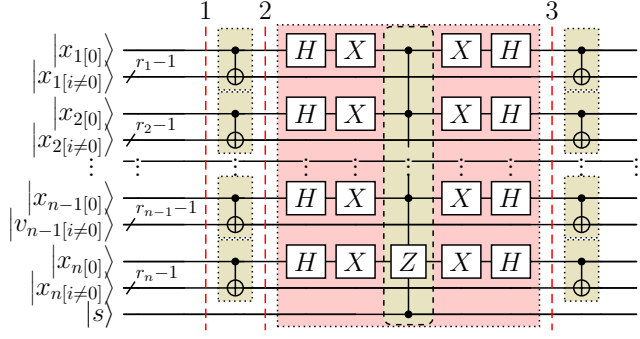

The circuit of the controlled diffuser is shown in Fig.\ref{fig:ParallelDiffuserWithCtrl}.  Unlike the traditional diffuser, the first step first involves disentangling distinct qubits across all variable qubit groups; this step is designed to correctly amplify the target states marked by the controlled oracle. In the second step, all variable qubits each successively pass through a Hadamard gate and an X gate. They then pass through a multi-controlled Z gate, which is also controlled by the control qubit. Subsequently, all variable qubits each successively pass through an X gate and a Hadamard gate. Finally, their entanglement is re-established within the variable qubit groups. After applying the controlled oracle and the controlled diffuser respectively, we then measure the control qubit. If the measurement result is 0, the evolution is halted. At this point, measure any one qubit in each variable qubit group; combining these distinct measurement results yields a target data. If the measurement result is 1, update $\phi_t$ and continue the iteration. The update rule for $\phi_t$ will be introduced later in the text.

The Theorem.\ref{theorem1} demonstrates the correctness of the algorithm.

\begin{theorem}\label{theorem1}If a SAT formula has at least one solution, then the PDFP search algorithm is guaranteed to correctly output the uniform superposition state of the target states. The update rule for $\phi_t$ can be chosen as \(\cos\phi_t = \frac{1 - \sin(\pi/(2t))}{1 + \sin(\pi/(2t))}\).\end{theorem}

\begin{proof}It can be seen that executing the controlled oracle circuit and the controlled diffuser circuit in sequence is equivalent to implementing the following unitary matrix:

\begin{equation}\begin{split}\label{proof1}
 U_{\text{ctrl}} &= (G\frac{I_2-S_z}{2}+\frac{I_2+S_z}{2})e^{i\phi_tS_y/2}(O\frac{I_2-S_z}{2}+\frac{I_2+S_z}{2})e^{-i\phi_tS_y/2}\\
 &=(G\frac{I_2-S_z}{2}+\frac{I_2+S_z}{2})(e^{-i\phi_tS_y}\frac{I_n-O}{2}+\frac{I_n+O}{2}).
 \end{split}
\end{equation}

To proceed from the first line to the second, we utilize the properties \(S_z^2 = I_2\) and \(O^2 = I_v\), where \(I_2\) denotes the identity matrix in the control qubit register and \(I_v\) denotes the identity matrix in the variable qubit register, respectively. Eqn.\ref{proof1} shows that sequentially applying the controlled oracle circuit and the controlled diffuser circuit is equivalent to the following process: first, a rotation of the control qubit is performed, conditioned on the oracle—if \(O|\psi_{\text{v}}\rangle = -|\psi_{\text{v}}\rangle\) (with \(|\psi_{\text{v}}\rangle\) representing the quantum state of the variable register), the control qubit is flipped; otherwise, it remains unchanged. This is followed by a $GO$ evolution conditioned on the control qubit. The overall effect is that as the system approaches the target state, the control qubit gradually ceases to induce the $GO$ operation.Each time a measurement is performed without recording the result \cite{mizel2009critically}, the density matrix of the total system can be expressed as \((1-\text{Tr}(\rho))|\psi_v\rangle\langle \psi_v| \otimes |0\rangle\langle 0|_s + \rho_v \otimes |1\rangle\langle 1|_s\), where \(\rho_v\) denotes the quantum state of the variable register. After each iteration, \(\rho_v\) is updated to:

\[\left[\begin{array}{c} \mathrm{Tr}\left(\rho_v' A\right) \\ \mathrm{Tr}\left(\rho_v' O\right) \\ \mathrm{Tr}\left(\rho_v'\right) \end{array}\right]=\left[\begin{array}{ccc} \cos 2\theta \cos\phi_t & \sin 2\theta \cdot \frac{1+\cos^2\phi_t}{2} & \sin 2\theta \cdot \frac{1-\cos^2\phi_t}{2} \\ -\sin 2\theta \cos\phi_t & \cos 2\theta \cdot \frac{1+\cos^2\phi_t}{2} & \cos 2\theta \cdot \frac{1-\cos^2\phi_t}{2} \\ 0 & \frac{1-\cos^2\phi_t}{2} & \frac{1+\cos^2\phi_t}{2} \end{array}\right]\left[\begin{array}{c} \mathrm{Tr}(\rho_v A) \\ \mathrm{Tr}(\rho_v O) \\ \mathrm{Tr}(\rho_v) \end{array}\label{proof2}\right]\]

There is a critical $\phi_t$ defined by \(\cos\phi_t = \frac{1 - \sin\theta}{1 + \sin\theta}\), where the three eigenvalues of the matrix $E_t$ in Eqn.\ref{proof2} all equal \(\cos\phi_t\). In this circumstance, the three components \(\text{Tr}(\rho_v A)\), \(\text{Tr}(\rho_v O)\), and \(\text{Tr}(\rho_v)\) all tend to diminish as the iterations increase. Consequently, the probability of finding the target state increases with time. When $M$ is unknown, the value of \(\phi\) corresponding to critical damping remains unspecified. Nevertheless, for iterations where \(n > 1\), one can select \(\cos\phi_t = \frac{1 - \sin(\pi/(2t))}{1 + \sin(\pi/(2t))}\) \cite{mizel2009critically}. When \( n \) is relatively small, the search is nearly classical; when \( n \) is large, the search reverts to the original Grover search. The advantage of choosing such an update rule for \( \phi_t \) lies in the fact that if the number of solutions is comparable to \( N \), the initial few iterations ensure a high probability of finding a solution. Of course, this is not the only update rule, and the rules themselves are even worthy of further investigation. We next prove that under such an update rule, the algorithm will eventually obtain the correct solution.

The first step of the proof lies in the fact that the magnitudes of all eigenvalues are less than $1$, and that the magnitudes of the eigenvalues equal $1$ if and only if \(\phi = 0\)\cite{mizel2009critically}, i.e., when the control qubit ceases to function. the spectral radius $R(E_t)$ is defined as the maximum magnitude of all eigenvalues of $E_t$,

\begin{equation}R(E_t) = \max_{\lambda \in \sigma(E_t)} |\lambda|,\end{equation}

where $\sigma(E_t)$ denotes the set of eigenvalues of $E_t$. For the update rule \(\cos\phi_t = \frac{1 - \sin(\pi/(2t))}{1 + \sin(\pi/(2t))}\), $R(E_t)\leq1$, where the equality holds as $t\rightarrow\infty$. By the submultiplicative property of the spectral radius, $R(E_1E_2E_3...E_t)\leq R(E_1)R(E_2)R(E_3)...R(E_t)$, $\lim_{t\rightarrow\infty}R(E_1)R(E_2)R(E_3)...R(E_t)=0 $. For any matrix \( E\) and \(\epsilon > 0\), there exists an operator norm \(\|\cdot\|_\epsilon\) such that \(\|E\|_\epsilon \leq R(E) + \epsilon\) \cite{horn2012matrix}. Consider $u_t=E_1E_2E_3...E_tu_0$, it satisfies
\begin{equation}\begin{split}|| u_t||_\epsilon&=||E_1E_2E_3...E_tu_0||_\epsilon\\&\leq||E_1E_2E_3...E_t||_\epsilon \cdot||u_0||_\epsilon\\&\leq (R(E_1E_2E_3...E_t)+\epsilon)\cdot||u_0||_\epsilon.\end{split}\end{equation}
As $t\rightarrow\infty$, $\lim_{t\rightarrow\infty}R(E_1E_2E_3...E_t)=0$.By the arbitrariness of \(\epsilon\), \(\lim_{t \to \infty} \|\boldsymbol{u}_t\|_\epsilon = 0\). \(u_t\) is defined as a vector encoding non-target state components. Thus, the algorithm will flip the controlled qubit deterministically, and we obtain the target state.

\end{proof}

We now finalize the proof of the algorithm's convergence. Furthermore, numerical results show that the number of queries required by the algorithm is at most $1.5$ times that of Grover's algorithm \cite{mizel2009critically}. 

The query complexity of the algorithm remains \(O(\sqrt{N})\). Thus, with relatively low overhead, we can implement a fixed-point quantum search algorithm when the number of solutions is unknown. This algorithm is clearly superior to the approach of first performing quantum counting and then executing Grover's algorithm. Furthermore, due to the parallelism of the oracle, the overall runtime of the algorithm can be significantly reduced. However, in the current era, the number of available qubits in each quantum computer is limited; therefore, it is natural to consider computing the algorithm distributively across different computers. So, how can this algorithm be implemented in a distributed manner?

\subsection{Distributed computing based on teleportation}

First, we note that only some multi-controlled gates require distributed processing; these gates are marked in olive in Figs.\ref{fig:ParallelOracleControlled} and \ref{fig:ParallelDiffuserWithCtrl}. In this paper, we consider a scheme based on teleportation \cite{sarvaghad2021general,lin2024parallel}. The procedure is shown in Fig.\ref{fig:mCUgate}. To perform multi-controlled gate operations in a distributed manner, we first prepare a large number of Bell pairs \(|e_i\rangle|\tilde{e}_i\rangle = \frac{1}{\sqrt{2}}(|00\rangle+|11\rangle)\), and distribute the qubits \(e_i\) and \(\tilde{e}_i\) to different sub-nodes and the master node, respectively. Then we proceed as follows:

In the first step, each node \(i\) applies a $CNOT$ operation to the qubit pair \(|p_i\rangle|e_i\rangle\), then measures Pauli-$Z$ of \(|e_i\rangle\). The result of this measurement is transmitted to the master node through a classical communication channel. The master node applies an $X$ gate to \(|\tilde{e}_i\rangle\) according to the measurement result is whether \(|1\rangle\) or not; if not, no operation is performed. After measurement, the state of \(e_i\) is deterministic, and the two qubits \(|p_i\rangle|\tilde{e}_i\rangle\) become entangled, taking the form of \(\gamma_i|00\rangle+\lambda_i|11\rangle\).  

In the second step, given that \(|p_i\rangle\) and \(|\tilde{e}_i\rangle\) now share the same state due to entanglement, using \(\tilde{e}_i\) as the i-th control qubit for the controlled-U gate is equivalent to using \(p_i\) as the i-th control qubits.  

In the third step, we disentangle \(p_i\) and \(\tilde{e}_i\). To achieve this, the master node measures Pauli-$X$ of the \(\tilde{e}_i\).  Then sends this measurement result to node \(i\) via a classical channel. Upon receiving the result, we apply a $Z$ gate to \(p_i\) if the outcome is \(|-\rangle\); otherwise, no operation is performed. 

Once each node completes these steps, the controlled-U gate operation is successfully executed in a distributed manner across nodes.
The proof of the correctness for this scheme can find in \cite{sarvaghad2021general,lin2024parallel}. Thus, each sub-node computes only part of the clauses and returns partial results to the client. This distributed protocol can protect the client’s privacy to a certain extent.

\begin{figure}[tb]
\centering
\scalebox{0.65}{
\tikzset{
my label/.append style={above right,xshift=1mm}
}
\begin{quantikz}[row sep={6mm,between origins}, column sep=5mm, font=\Large]
\lstick{\ket{p_1}} & [-3mm] \ctrl{1} & [-3mm] \qw & [-5mm] \qw & [-3mm] \qw \slice{1} & \qw \slice{2} & [2mm] \gate{Z} & [-6mm] \qw & [-5mm] \qw \slice{3} & \qw \\
\lstick{\ket{e_1}} & \targ{} & \meter{0/1} \vcw{7} & \qw & \qw & \qw & \qw & \qw & \qw & \qw \\[-1mm]
\wave&&&&&&&&&\\[-1mm]
\lstick{\vdots} & & & & \vdots & & & & \vdots & \\[-1mm] 
\wave&&&&&&&&&\\
\lstick{\ket{p_m}} & \qw & \qw & \ctrl{1} & \qw & \qw & \qw & \qw & \gate{Z} & \qw \\
\lstick{\ket{e_m}} & \qw & \qw & \targ{} & \meter{0/1} \vcw{4} & \qw & \qw & \qw & \qw & \qw \\
\wave&&&&&&&&& \\[1mm]
\lstick{\ket{\Tilde{e}_1}} & \qw & \gate{X} & \qw & \qw & \ctrl{1} \gategroup[4,steps=1,style={dashed,
rounded corners,fill=olive!20, inner xsep=0.1pt}, background,label style={label position=below,anchor=north,yshift=-2mm}]{{}} & \meter[style={label position=below right}]{\ket{\pm}} \vcw{-8} & \qw & \qw & \qw \\
\hspace{-8mm} \vdots & & & \cdots  & & \vdots & & \cdots & & \\
\lstick{\ket{\Tilde{e}_m}} & \qw & \qw & \qw & \gate{X} & \ctrl{1} \vqw{-1} & \qw & \qw & \meter[style={label position=below right}]{\ket{\pm}} \vcw{-5} & \qw \\
\lstick{\ket{t}} & \qw & \qw & \qw & \qw & \gate{U} & \qw & \qw & \qw & \rstick{\ket{t'}} \qw
\end{quantikz}
\hspace{-2mm}
=
\begin{quantikz}[row sep={7mm,between origins}, column sep=2mm, font=\Large]
\lstick{\ket{p_1}} & \ctrl{1} \gategroup[4,steps=1,style={dashed,
rounded corners,fill=olive!20, inner xsep=0.1pt}, background,label style={label position=below,anchor=north,yshift=-2mm}]{{}} & \qw \\
\hspace{-9mm} \vdots & \vdots \\
\lstick{\ket{p_m}} & \ctrl{1} \vqw{-1} & \qw \\
\lstick{\ket{t}} & \gate{U} & \qw
\end{quantikz}
}
\caption{The scheme for the distributed $m$-controlled-$U$ gate.}\label{fig:mCUgate}
\end{figure}
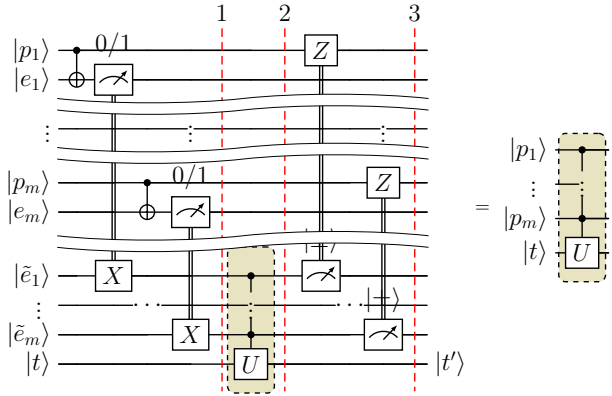

\subsection{Numerical result}

To test the effectiveness of our algorithm, we apply it to solving a simple SAT formula, which has a unique solution: \(a=1\), \(b=1\), \(c=1\). For comparison, we first solve the problem using Grover's algorithm. It can be observed that Grover's algorithm can obtain the solution with high probability in only two steps, and the probability of finding the solution approaches 1 at the sixth step. However, as anticipated, the success probability exhibits oscillatory behavior with the number of iterations.Next, we solve the problem using the PFP search algorithm. We find that whether using fixed critical \(\phi\) values or updated \(\phi\) values (which are applicable to scenarios where the number of solutions is known and unknown, respectively), Grover's algorithm outperforms PFP in the first two steps. Nevertheless, as the number of steps increases, the success probability of the PFP algorithm increases monotonically and quickly becomes sufficiently close to 1. Thus, this algorithm successfully overcomes the soufflé problem.
\begin{figure}    \centering    \includegraphics[width=1\linewidth]{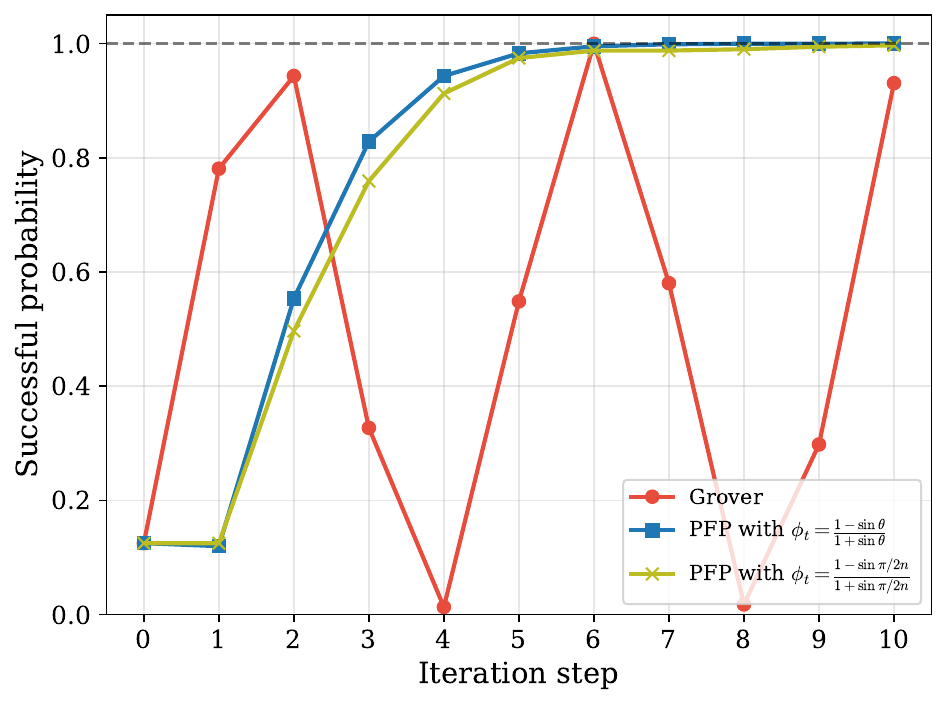}    \caption{Numerical results for solving SAT formula $(a)\land(a\lor b)\land(a\lor c)$ with Grover's algorithm and PFP search algorithm.}    \label{fig:numercialresult}\end{figure}

\section{Summary}\label{Summary}

In this paper, we propose a parallel fixed-point quantum search algorithm for solving the SAT problem. This algorithm guarantees the deterministic discovery of solutions while retaining quadratic speedup. Furthermore, parallel computation can linearly reduce the total runtime, making the algorithm well-suited for the current NISQ era. We also discuss how to execute this algorithm in a distributed manner, which will further address the issue of insufficient available qubits in a single quantum computer during the NISQ era. Future research directions include validating the algorithm on actual quantum hardware and designing more effective \(\phi_t\) update rules.

\bibliographystyle{plainnat}
\bibliography{mybibliography}

\end{document}